\documentclass[12pt]{article}
\usepackage{amsfonts,amssymb,amsbsy,amsmath,amsthm,mathrsfs,enumerate,verbatim}
\usepackage[mathscr]{eucal}

 \usepackage[colorlinks=true]{hyperref}
 \usepackage{graphicx}

 \hypersetup{urlcolor=blue, citecolor=red}
\newtheorem{theo}{Theorem}[section]
\newtheorem{prop}[theo]{Proposition}

\newtheorem{defi}[theo]{Definition}
\theoremstyle{definition}
\newtheorem{rema}[theo]{Remark}

\newcommand{\bel}{\begin{equation} \label}
\newcommand{\ee}{\end{equation}}

\newcommand{\R}{{\mathbb R}}

\def\beq{\begin{equation}}
\def\eeq{\end{equation}}
\newcommand{\bea}{\begin{eqnarray}}
\newcommand{\eea}{\end{eqnarray}}
\newcommand{\beas}{\begin{eqnarray*}}
\newcommand{\eeas}{\end{eqnarray*}}

{

\begin{document}

\begin{center}
{\Large \bf Combinatorial invariant of nearly integrable Hamiltonians}
% Spectral monodromy of small non-selfadjoint perturbations of nearly integrable Hamiltonians
% Spectral monodromy in nearly integrable quantum systems
% Spectral monodromy of small non-selfadjoint perturbations of nearly integrable Hamiltonians
% Spectral monodromy of non-selfadjoint quantum Hamiltonians which are classically nearly integrable
\medskip

\today

\end{center}

\medskip

\begin{center}
Quang Sang PHAN

FITA, Vietnam National University of Agriculture, Hanoi, Vietnam \&

Vietnam Institute for Advanced Study in Mathematics, Hanoi, Vietnam

E-mail: pqsang@vnua.edu.vn
\end{center}

\begin{abstract} ~\\
We work with small non-selfadjoint perturbations of a selfadjoint quantum Hamiltonian with two degrees of freedom, assuming that the principal symbol of the selfadjoint part is (classically) a nearly integrable system, together with a globally non-degenerate condition.
We define a monodromy directly from the spectrum of such an operator, in the semiclassical limit.
Moreover, this spectral monodromy allows to detect the topological modification on the dynamics of the nearly integrable system.
It can be identified with the monodromy for KAM invariant tori of the nearly integrable system.
\end{abstract}

\medskip

% {\bf  AMS 2010 Mathematics Subject Classification:} 35R30.\\
% Classification: 35S05, 81Q20, 81Q12, 62M15, 37J35, 70H06, 34L20, 58C40, 78M35, 34M35

~\\
{\bf  Keywords:} Hamiltonian systems, monodromy, non-selfadjoint, asymptotic spectral, pseudo-differential operators, KAM theory.

\tableofcontents

%%%%%%%%%%%%%%%%%%%%%%%%%%%%%%%%%%%%%%%%%%%%%%%%%%%%%%%%%%%%%%%%%%%%%%%%%%%%%%%%%%%%%%%%%%%%%%%%%%%%%%%%%%%%%%%%%%%%%%%%%%%%%%%%%%%%%%%%%%%%%%%%%%%%%%%%%%%%%%%%%%%%%%%%%%%%%%%%%%%
%%%%%%%%%%%%%%%%%%%%%%%%%%%%%%%%%%%%%%%%%%%%%%%%%%%%%%%%%%%%%%%%%%%%%%%%%%%%%%%%%%%%%%%%%%%%%%%%%%%%%%%%%%%%%%%%%%%%%%%%%%%%%%%%%%%%%%%%%%%%%%%%%%%%%%%%%%%%%%%%%%%%%%%%%%%%%%%%%%%
%%%%%%%%%%%%%%%%%%%%%%%%%%%%%%%%%%%%%%%%%%%%%%%%%%%%%%%%%%%%%%%%%%%%%%%%%%%%%%%%%%%%%%%%%%%%%%%%%%%%%%%%%%%%%%%%%%%%%%%%%%%%%%%%%%%%%%%%%%%%%%%%%%%%%%%%%%%%%%%%%%%%%%%%%%%%%%%%%%%

\section{Introduction}

% définir une propriété de monodromie spectrale associée au spectre, vu comme un ensemble discret de points de C.

%   a property of monodromy associated spectral spectrum, seen as a discrete set of points of C.

We propose in this work a combinatorial invariant (monodromy) of a quantum Hamiltonian which is (classically) nearly integrable, by looking at the spectrum of non-selfadjoint perturbations.

Intuitively when we see a $2D$ lattice with (deformed) rectangular meshes (for example a fishing net or a fishnet shirt...), we propose a mysterious question whether these meshes can be regularly connected each to others, in the sense that certain mesh can be moved on the lattice to an arbitrary mesh? The answer is not evident for any lattice, which may be not a square lattice.
If the answer is positive, we say that there is not monodromy. Otherwise, we say that there exists a monodromy on the lattice.

Such a lattice structure can be appeared in many different mathematical objects. We knew an asymptotic lattice that is used to model the
 joint spectrum of system of commuting selfadjoint operators Refs. \cite{Char88}, and \cite{Cush88}. The quantum monodromy for this lattice was
completely defined in \cite{Vu-Ngoc99}.
For a non-selfadjoint quantum Hamiltonian, in certain assumption, its discrete spectrum may be locally a deformed lattice.
Whether a monodromy can be defined for the discrete spectrum of one single semiclassical operator?

We focus on non-selfadjoint $h-$Weyl-pseudodifferential operators, which are small non-selfadjoint perturbations of selfadjoint operators, with two degrees of freedom.

The first answer in this direction was given in Ref. \cite{QS14}, in which the operators had the simple form $P+ i  \varepsilon Q$, such that the corresponding principal symbols $p$ of $P$ and $q $ of $Q$ commute for the Poisson bracket, $\{p,q\}=0$.
And then, in the recent work Ref. \cite{QS16.1} we treated perturbed operators of the form $  P_{\varepsilon}= P(x,hD_x,\varepsilon; h )$ such that the principal symbol $p$ of the selfadjoint unperturbed operator $P_{\varepsilon=0}$ is a completely integrable Hamiltonian.

In the present work, non-selfadjoint operators of interest are related to nearly integrable Hamiltonian systems.
Such an operator is of the form
\begin{equation} \label{tt2}
P_{\varepsilon, \lambda}=P_{\lambda} (x,hD_x,\varepsilon; h ),
\end{equation}
depending on a small parameter $\varepsilon$ in the regime $h \ll \varepsilon = \mathcal{O}(h^\delta)$, with $0< \delta <1 $, and also smoothly on a small enough parameter $0< \lambda \ll 1$, such that the principal symbol $p_\lambda $ of the unperturbed selfadjoint operator is nearly integrable and globally non-degenerate. As an application of the spectral asymptotic theory Refs. \cite{Hitrik08}, \cite{Hitrik07} and with the help of the KAM theory Refs. \cite{HB90}-\cite{Poschel01} we will show in Section \ref{end mono} that, under some suitable general assumption, the spectrum of $P_{\varepsilon, \lambda}$ has locally the form of a deformed discrete lattice. Moreover, we can prove in Theorem \ref{mtheo2} that this spectrum is an asymptotic pseudo-lattice (see Definition \ref{pseu-lattice}).
Hence, applying a result from Ref. \cite{QS14}, we can define a combinatorial invariant- the \emph{spectral monodromy}, directly from the spectrum of $P_{\varepsilon, \lambda}$.

In aspect of the classical theory, it is known from Ref. \cite{Broer07} that there exists a geometric invariant (like monodromy) of the foliation by KAM invariant tori for the nearly integrable Hamiltonian system $p_\lambda $.
The spectral monodromy that we define in this work is the complete answer for an open question mentioned in Ref. \cite{Broer07} about the existence of quantum monodromy in nearly integrable cases. It explains certain defects in the quantum spectrum Refs. \cite{Cush04}, and \cite{Cush00}.
Moreover we shall prove in Theorem \ref{reco} that the spectral monodromy of $P_{\varepsilon, \lambda}$ allows to recover the monodromy of $p_\lambda $.
It means that the spectral monodromy proposes a way of detecting the topological modification (angle-action coordinates) on the dynamics of nearly integrable Hamiltonian systems associated to non-selfadjoint quantum Hamiltonians from only one spectrum.

%%%%%%%%%%%%%%%%%%%%%%%%%%%%%%%%%%%%%%%%%%%%%%%%%%%%%%%%%%%%%%%%%%%%%%%%%%%%%%%%%%%%%%%%%%%%%%%%%%%%%%%%%%%
%%%%%%%%%%%%%%%%%%%%%%%%%%%%%%%%%%%%%%%%%%%%%%%%%%%%%%%%%%%%%%%%%%%%%%%%%%%%%%%%%%%%%%%%%%%%%%%%%%%%%%%%%%%%%%
\section{Preparation}

\subsection{Classical theory}

We recall here some notions and results from the classical theory.

\begin{defi} \label{ints}
 An integrable system on a symplectic manifold $(W, \sigma)$ of dimension $ 2n $ ($n \geq 1$) is given $n$ smooth real-valued functions $f_1, \dots, f_n$ in involution with respect to the Poisson bracket generated from the symplectic form $ \sigma$, whose differentials are almost everywhere linearly independent.
In this case, the map $$ F=(f_1, \dots, f_n ): W \rightarrow \mathbb R^n $$ is also called an \textit{integrable system}.

A smooth function $f_1$ is called completely integrable if there exists $n-1$ functions $f_2, \dots, f_n$ such that $F=(f_1, \dots, f_n )$ is an integrable system.
\end{defi}

Let $U$ be an open subset of regular values of $F$. Then we have,

\begin{theo}[Angle-action theorem] (Refs. \cite{Arnold67}, and \cite{Duis80})  \label{A-A}
Let $c \in U$, and $\Lambda_c$ be a compact regular leaf of the fiber $F^{-1}(c)$. Then there exists an open neighborhood $V= V^c$ of $\Lambda_c$ in $W$ such that $F\mid_ {V} $ defines a smooth locally trivial fibre bundle onto an open neighborhood $ U^c \subset U$ of $c$, whose fibres are invariant Lagrangian $n-$tori. Moreover, there exists a symplectic diffeomorphism $\kappa = \kappa^c$,
    $$\kappa= (x,\xi): V \rightarrow \mathbb T ^n \times  A, $$
with $A= A^c \subset \mathbb R^n$ is an open subset, such that $F\circ \kappa^{-1}(x, \xi)= \varphi(\xi)$ for all $x \in \mathbb T^n $,
and $\xi \in A$, and here $\varphi= \varphi^c: A \rightarrow \varphi (A)=U^c $ is a local diffeomorphism. We call $( x, \xi)$ local angle-action variables near $\Lambda_c$ and $(V, \kappa)$ a local angle-action chart.
\end{theo}
% $x=(x_1, \dots, x_n) \in \mathbb T^n $,  and $\xi= ( \xi _1, \dots \xi_ n ) \in A$,

Notice that one chooses usually the local chart such that the torus $\Lambda_c$ is sent by $\kappa$ to the zero section $T ^n \times \{0 \}$. By this theorem, for every $a \in U^c $, then $\Lambda_a:= F^{-1}(a) \cap V^c$ is an invariant Lagrangian $n-$torus, called a Liouville torus, and we write
   \begin{equation} \label{tori} \Lambda_a \simeq \kappa(\Lambda_a)=\mathbb T ^n \times \{ \xi_a \}:= \Lambda_{\xi_a}, \end{equation}
with some $\xi_a \in A$.

%%%%%%%%%%%%%%%%%%%%%%%%%%%%%%%%%%%%%%%%%%%%%%%%%%%%%%%%%%%%%%%%%%

\subsection{The monodromy of an asymptotic pseudo-lattice} \label{mapl}
We recall here the definition of an asymptotic pseudo-lattice and its monodromy given in Ref. \cite{QS14}. This is a discrete subset of $\mathbb R^2$ admitting a particular property.

We shall see later that the spectrum should be included in a horizontal band of size $\mathcal{O}(\varepsilon)$. This suggests us
introducing the function
        \begin{eqnarray}   \label{chi}
          \chi :  \mathbb R^2  & \rightarrow&  \mathbb R^2  \cong  \mathbb C
            \\
            u= (u_1,u_2)  &  \mapsto & \chi_u= \ (u_1, \varepsilon u_2) \cong u_1+i \varepsilon u_2,   \nonumber
          \end{eqnarray}
in which we identify $\mathbb C$ with $\mathbb R^2$.

For any subset $U$ of $\mathbb R^2$ we denote
\beq \label{elp}
 \ U(\varepsilon)=  \chi(U).
\eeq

\begin{defi} \label{pseu-lattice}
     Let $U$ be an open subset of $\mathbb R^2$ with compact closure and
     let $\Sigma(\varepsilon, h)$ (which depends on small $h$ and $\varepsilon$) be a discrete subset of $U(\varepsilon)$.
     For $h, \varepsilon$ small enough and in the regime $ h \ll \varepsilon$, we say that $(\Sigma(\varepsilon, h), U(\varepsilon))$
     is an asymptotic pseudo-lattice if:
         for any small parameter $\alpha >0$, there exists a set of good values in $\mathbb R^2$, denoted by $\mathcal{G}(\alpha)$, whose complement is of size $\mathcal O(\alpha)$ in the sense:
             $$\mid {}^C \mathcal{G} (\alpha) \cap I  \mid \leq C \alpha \mid I \mid, $$
with a constant $C>0$ for any domain $I \subset \mathbb R^2$; for all $c \in U$, there exists a small open subset $U^c \subset U $ around $c$ such that for every good value $a \in U^c  \cap \mathcal{G}(\alpha)$, there is an adapted good rectangle $R^{(a)}(\varepsilon,h) \subset U^c( \varepsilon)$ of the form \eqref{cua so} (with $a=(E,G)$), and a smooth local diffeomorphism $f= f(\cdot, \varepsilon;h)$ which sends $R^{(a)}(\varepsilon,h)$ on its image, satisfying
\beq  \label{semi-cart}
    \Sigma( \varepsilon, h) \cap R^{(a)}(\varepsilon,h) \ni  \mu \mapsto  f(\mu,\varepsilon; h) \in h \mathbb Z^2 +\mathcal O(h^\infty).                        \eeq
Moreover, the function $\widetilde{f}:= f \circ \chi$, with $\chi$ defined by \eqref{chi}, admits an asymptotic expansion in
$(\varepsilon,\frac{h}{\varepsilon})$ for the $C ^\infty-$topology in a neighborhood of $a$,
uniformly with respect to the parameters $h$ and $\varepsilon$, such that its leading term $\widetilde{f}_0$
is a diffeomorphism, independent of $\alpha$, locally defined on the whole $U^c $ and independent
of the selected good values $a \in U^c $.

We also say that the couple $(f(\cdot, \varepsilon;h), R^{(a)}(\varepsilon,h) )$ is a $h-$chart of $\Sigma( \varepsilon, h)$,
and the family of $h-$charts $(f(\cdot, \varepsilon;h), R^{(a)}(\varepsilon,h) )$, with all $a
\in U^c \cap \mathcal{G}(\alpha) $, is a local pseudo-chart on $U^c( \varepsilon)$ of
$(\Sigma(\varepsilon, h), U(\varepsilon))$.
\end{defi}

% Ban luan ve lattice nay? miniscopine???????????

% \begin{rema} A standard lattice is obviously an asymptotic pseudo-lattice. An another example is a known lattice with some similar but lighter properties, named asymptotic lattice, given in Ref. \cite{Vu-Ngoc99}. That lattice is modeled on the joint spectrum of system of commuting operators. It is locally defined, while the asymptotic pseudo-lattice is very delicate, it is $h-$locally defined.

% The introduction of discrete lattices aims to show that the combinatorial invariant that we will define is directly built from the spectrum of operators. If different operators have the same spectrum, then they have the same monodromy.
% \end{rema}

%%%%%%%%%%%%%%%%%%%%%%%%%

Let $\{ U^j \}_{j \in \mathcal{J} }$, here $\mathcal{J}$ is a finite index set, be an arbitrary (small enough) locally finite covering of $U$. Then the asymptotic pseudo-lattice $(\Sigma(\varepsilon, h), U(\varepsilon))$ can be covered by associated local pseudo-charts $ \{ \left(  f_ j (\cdot, \varepsilon; h),  U^j (\varepsilon) \right) \}_{j \in \mathcal{J} }$. Note from Definition \ref{pseu-lattice} that the leading terms $ \widetilde{f}_{j,0}(\cdot, \varepsilon; h)$ are well defined on whole $ U^j$ and we can see them as
the charts of $U$. Analyzing transition maps, we had the following result.
\begin{prop} [Ref. \cite{QS14}]   \label{dl}
 On each nonempty intersection $  U^i \cap U^j \neq \emptyset$, $i, j \in \mathcal{J}$, there exists an unique integer linear map $M_{ij} \in GL(2,
\mathbb Z)$ (independent of $h, \varepsilon$) such that:
\begin{equation} \label{transition-pseu}
d \big ( \widetilde{f}_{i, 0} \circ (\widetilde{f}_{j, 0})^{-1} \big )= M_{ij}.
\end{equation}
\end{prop}

Then we define the (linear) monodromy of the asymptotic pseudo-lattice $(\Sigma(\varepsilon, h), U(\varepsilon))$ as the $1$-cocycle of $\{M_{ij} \} $, modulo coboundary, in the \v{C}ech cohomology of $U$ with values in the integer linear group $GL(2, \mathbb Z)$,
denoted by
         \begin{equation} \label{mono}
             [\mathcal M] \in \check{H}^1(U,GL(2, \mathbb Z) ).
         \end{equation}
It doesn't depend on the selected finite covering $\{ U^j \}_{j \in \mathcal{J} }$ and the small parameters $h, \varepsilon$.

We can also associate the class $[\mathcal M]$ with its holonomy, that is a group morphism from the fundamental group $\pi_1(U)$ to the group $ GL(2, \mathbb Z)$, modulo conjugation. The triviality of $[\mathcal M]$ is equivalent to the one of its holonomy.

% For more detail of this point, we refer to Ref. \cite{QS14}.

%%%%%%%%%%%%%%%%%%%%%%%%%%%%%%%%%%%%%%%%%%%%%%%%%%%%%%%%%%%%%%%%%%%%%%%%%%%%%%%%%%%%%%%%%%%
\subsection{Weyl-quantization}  \label{sec2.1}
We will work throughout this article with pseudodifferential operators obtained by the $h-$Weyl-quantization of a standard space of symbols on $T^*M =\mathbb R^{2n}_{(x,\xi)}$, here $M= \mathbb R^n$ or a compact manifold of $n$ dimensions, and in particular $n=2$. We denote $\sigma $ the standard symplectic $2-$form on $T^*M$.

 Note that if $ M = \mathbb R ^ n $, the volume form $ \mu (dx) $ is naturally induced by the Lebesgue measure on $ \mathbb R ^ n $.
    If $ M $ is a compact Riemannian manifold, then the volume form $ \mu(dx) $ is induced by the given Riemannian structure of $ M $.
    Therefore in both cases the volume form is well defined.

In the following we represent the quantization in the case $M= \mathbb R^n$. In the manifold case, the quantization is introduced suitably.
We refer to Refs. \cite{Dimas99}-\cite{Shubin01} for the theory of pseudodifferential operators.

        \begin{defi} \label{fonc ord}
                        A function $m: \mathbb R^{2n} \rightarrow (0, + \infty)$ is called an order function
                     if there are constants  $C,N >0$ such that
                                $$m(X)  \leq C \langle X-Y\rangle^{ N} m(Y), \forall X,Y \in \mathbb R^{2n},$$
        with notation $\langle Z\rangle= (1+ |Z|^2)^{1/2}$ for $Z \in \mathbb R^{2n}$.
        \end{defi}

% One use often the order function  $m(Z) \equiv 1$ or $$m(Z)= \langle Z \rangle ^{l/2}= (1 + |Z|^2 )^{l/2},$$ with a given constant $l \in \mathbb R $.

         \begin{defi}
                        Let $m$ be an order function and $k \in \mathbb R$, we define the class of symbols of $h$-order $k$, denoted by $S^k(m)$, of functions $(a(\cdot;h))_{h \in (0,1]}$ on $\mathbb R^{2n}_{(x,\xi)}$ by
                        \begin{equation}
                                S^k(m)= \{ a \in C^\infty (\mathbb R^{2n})
                                 \mid  \forall \alpha \in \mathbb N ^{2n}, \quad |\partial^\alpha a | \leq  C_\alpha h^k m \} ,
                        \end{equation}
         for some constant $C _\alpha >0$, uniformly in $h \in (0,1]$. \\
         A symbol is called $\mathcal O(h^\infty)$ if it's in $\cap _{k \in \mathbb R } S^k(m):= S^{\infty}(m) $.
        \end{defi}

        Then $ \Psi^k(m)(M)$ denotes the set of all (in general unbounded) linear operators $A_h$ on $L^2(M, \mu(dx)) $, obtained from the $h-$Weyl-quantization of symbols $a(\cdot;h) \in S^k(m) $ by the integral operator
        \begin{equation} \label{symbole de W}
                            (A_h u)(x)=(Op^w_h (a) u)(x)= \frac{1}{(2 \pi h)^n}
                                 \int_{ \mathbb R^{2n}} e^{\frac{i}{h}(x-y)\xi}
                                 a(\frac{x+y}{2},\xi;h) u(y) dy d\xi.
        \end{equation}

In this work, we always assume that symbols admit classically asymptotic expansions in integer powers of $h$.
\begin{equation}
        a(x, \xi; h ) \sim \sum_{j=0}^\infty
                             a_{j}(x,\xi) h^j, h \rightarrow 0,
     \end{equation}

The leading term in this expansion is called the principal symbol of operators.

%%%%%%%%%%%%%%%%%%%%%%%%%%%%%%%%%%%%%%%%%%%%%%%%%%%%%%%%%%%%%%%%%%%%%%%%%%%%%%%%%%%%%%%%%%%%%%%%%%%%%%%%%%%%%%%%%%%%%%%%%%%%
%%%%%%%%%%%%%%%%%%%%%%%%%%%%%%%%%%%%%%%%%%%%%%%%%%%%%%%%%%%%%%%%%%%%%%%%%%%%%%%%%%%%%%%%%%%%%%%%%%%%%%%%%%%%%%%%%%%%%%%%%%%%%%%

%%%%%%%%%%%%%%%%%%%%%%%%%%%%%%%%%%%%%%%%%%%%%%%%%%%%%%%%%%%%%%%%%%%%%
%%%%%%%%%%%%%%%%%%%%%%%%%%%%%%%%%%%%%%%%%%%%%%%%%%%%%%%%%%%%%%%%%%%%%

%%%%%%%%%%%%%%%%%%%%%%%%%%%%%%%%%%%%%%%%%%%%%%%%%%%%%%%%%%%%%%%%%%%%%%%%%%%%%%%%
%%%%%%%%%%%%%%%%%%%%%%%%%%%%%%%%%%%%%%%%%%%%%%%%%%%%%%%%%%%%%%%%%%%%%%%%%%%%%%%%

\section{Spectral monodromy in the nearly integrable case}

% Trang 2:  general assumption that the real energy space of the unperturbed leading symbol contains several flow invariant Lagrangian tori
% satisfying a Diophantine condition.

In this section we consider small non-selfadjoint perturbations of a selfadjoint semiclassical operator in $2$ dimension, whose leading symbol is close to a completely integrable non-degenerate system.

Together with some global assumptions on the dynamics as given in Section \ref{gnrass}, a radical assumption in the spectral asymptotic theory for such an operator from a $h-$depending complex window, given in Ref. \cite{Hitrik07}, is that the real energy surface at certain level of the unperturbed leading symbol possesses several Hamiltonian flow invariant Lagrangian tori, satisfying a Diophantine condition.
In the nearly integrable case, this may be archived with the help of KAM theory (by Kolmogorov-Arnold-Moser), see Refs. \cite{Kol67}, and \cite{HB90}-\cite{Poschel01}.

We will give spectral asymptotic results to show that the spectrum of such an operator should be an asymptotic pseudo-lattice. Then we can define its spectral monodromy by applying Section \ref{mapl}.

Through this section, $ M $ denotes $\mathbb R^2$ or a connected compact analytic real (Riemannian) manifold of dimension $ 2 $.

\subsection{Quasi-periodic flows of nearly integrable systems} \label{kams}

Classical KAM theory allows to treat perturbations of completely integrable Hamiltonian systems. Under a Kolmogorov condition, this theory proves locally the persistence of invariant Lagrangian tori, called KAM tori, such that on a KAM torus the classical flow of the unperturbed systems stills quasi-periodic with Diophantine constant frequency. Moreover, the union of these KAM tori is a nowhere dense set, with complement of small measure in the phase space.

% More precisely,

% To KAM, we refer to Refs.

We consider a perturbed Hamiltonian that is close to a completely integrable (non-degenerate) one
\beq  p_\lambda= p+ \lambda p_1, \ 0< \lambda \ll 1,  \eeq where $p$ and $p_1$ are holomorphic
bounded Hamiltonians in a tubular neighborhood of $T^*M$, real on  $T^*M$, and moreover $p$ is a completely integrable Hamiltonian system, according Definition \ref{ints}.

For some $E \in \R$, and suppose that $\Lambda_a \subset p^{-1}(E)$ is an $H_p-$invariant Diophantine Lagrangian torus.
It can be locally embedded in a Lagrangian foliation of $H_p-$invariant tori.
By Theorem \ref{A-A}, there is an angle-action local chart $(\kappa, V)$ near $\Lambda_a$
\begin{equation}  \label{coor} \kappa= (x,\xi): V \rightarrow \mathbb T ^2 \times  A, \end{equation}
such that the function $p$ becomes a function only of $\xi$,
           \begin{equation} \label{p}
                    p\circ \kappa^{-1} =p(\xi)= p(\xi_1, \xi_2), \ \xi \in A,         \end{equation}
and then
$$ p_\lambda \circ \kappa^{-1}= p_\lambda (x, \xi)=p (\xi)+ \lambda p_1 (x, \xi). $$
For any Liouville torus $\Lambda \subset V $ such that by $\kappa$,
\begin{equation}  \label{Lam}
\Lambda \simeq \Lambda_{\xi}, \xi \in A,
\end{equation}
we define the frequency of $\Lambda$ (also of $\Lambda_{\xi}$) by
\begin{equation} \label{frequence}
\omega(\xi)= \frac{\partial p}{\partial \xi} (\xi)= \big ( \frac{\partial p}{\partial \xi_1} (\xi),
 \frac{\partial p}{\partial \xi_2} (\xi) \big ), \ \xi \in A,
 \end{equation}
 and its rotation number by
 \begin{equation} \label{rho}
 \rho(\xi)=  \big [ \frac{\partial p}{\partial \xi_1} (\xi):
 \frac{\partial p}{\partial \xi_2} (\xi) \big ], \ \xi \in A ,
 \end{equation}
 viewed as an element of the real projective line.
We note that $\omega$ and $\rho$ depend analytically on $\xi \in A$.
% In particular, the frequency of $\Lambda_a \simeq  \Lambda_{\xi_a} $ is $\omega (\xi_a) = \frac{\partial p}{\partial \xi} (\xi_a) $ satisfying the Diophantine condition \eqref{dn alpha-ddioph}, for some $\alpha >0$.

Then the Hamiltonian flow of $p$ on $\Lambda $ is quasi-periodic of the constant Hamiltonian vector field
\beq \label{Hp}
 H_p(x, \xi)=  \omega_1 (\xi)
\frac{\partial }{\partial x_1}+  \omega_2(\xi)  \frac{\partial }{\partial x_2}, \ x \in \mathbb
T^2, \ \xi \in A. \eeq

\begin{defi} \label{Kol}
We say that $p$ is locally (Kolmogorov) non-degenerate (on $V$) if the local
frequency map $\omega: A \rightarrow \mathbb R^2$, given by \eqref{frequence}, is a
diffeomorphism onto its image.
\end{defi}
In fact, this condition is equivalent to
$$det [ \frac{\partial \omega}{\partial \xi}] = det [\frac{\partial^2 p}{\partial \xi^2}] \ \neq 0 \ \textrm{on} \ A,$$
and it means that the $H_p-$invariant Liouville tori in $V$ can (locally) be parametrized by
their frequencies.

\begin{defi}  \label{diop}
 Let $\alpha >0 $, $d>0$, and $\Lambda \simeq \Lambda_{\xi} $ be an invariant Lagrangian torus, as in \eqref{Lam}.
 We say that $\Lambda$ is $(\alpha,d)-$Diophantine if its frequency, defined in
 (\ref{frequence}), satisfies
 \begin{equation}  \label{dn alpha-d dioph}
     \omega(\xi) \in D_{\alpha,d}= \big \{   \omega \ \in \mathbb R^2 \big |  \   | \langle \omega,k \rangle |  \geq \frac{\alpha}{ |k|^{1+d}}, \
     \forall \ k \in \mathbb Z^2 \backslash \{ 0 \} \big \}.
 \end{equation}

\end{defi}
Notice also that when $d>0$ is fixed, the Diophantine property (for some $\alpha >0$) of $\Lambda$ is independent of selected angle-action coordinates, see Ref. \cite{Bost86}. If $\Lambda$ is $(\alpha,d)-$Diophantine, then its frequency must be irrational and the result in Remark \ref{rem1} is true.

It is known that the set $D_{\alpha,d} $ is a closed set with closed half-line structure. When we take
$\alpha$ to be sufficiently small, it is a nowhere dense set but with no isolated points. Moreover, its
measure tends to large measure as $\alpha$ tends to $0$, i.e., the measure of its complement is of order
$\mathcal{O}(\alpha)$. See Refs. \cite{HB90}, and \cite{Poschel01}.

Let $\Omega= \omega(A) $ be the open range of the frequency map $\omega$.
% Let us recall that the set $D_{\alpha, d}$ given in \eqref{dn alpha-d dioph} denoted the set of $(\alpha,d)-$Diophantine frequencies.
Let $\Omega_{\alpha, d} \subset \Omega$ be the subset of frequencies which satisfy the Diophantine condition \eqref{dn alpha-d dioph} and whose distance to the boundary of $\Omega$ is at least equal to $\alpha$.
It is known that the set $\Omega_{\alpha, d}$ is a nowhere dense perfect subset of large measure for sufficiently small $\alpha$.
The measure of $\Omega \setminus \Omega_{\alpha, d} $ is $\mathcal O(\alpha )$, which tends to zero as $\alpha \downarrow 0$, see Refs. \cite{Poschel01}, and \cite{HB90}.
Finally, we define the subset $$A_{\alpha, d}=  \omega^{-1} ( \Omega_{\alpha, d}) \subset A.$$
It is true that $A_{\alpha, d}$ is also a nowhere dense perfect subset of large measure. The measure of the complement subset $A \setminus A_{\alpha, d}$ is of order $\mathcal O(\alpha )$ as $\alpha \downarrow 0$, see Ref. \cite{Broer07}. The intersection of $p^{-1}(E)$ with $\mathbb T^2 \times A$ is of the form  $\mathbb T^2 \times \Gamma_a $, with a some curve denoted by $\Gamma_a$ in $A$, passing through $\xi_a$.

Now for $\alpha$ small enough, we have the quasi-periodic stability of the Diophantine invariant
Lagrangian tori in $\mathbb T^2 \times A_{\alpha, d} $, as the following theorem.
It is combined form the different known versions of the classical KAM theorem, see Refs. \cite{Poschel82}, \cite{Bost86}, and \cite{Poschel01}.

% In particular the paper cite 33 proved the smooth dependence, in the sense of Whitney, of the KAM tori on the frequencies.

\begin{theo} [Local KAM] \label{kam}

Assume that $p$ is locally non-degenerate as Definition \ref{Kol}. Let $d >0$ fixed and $\alpha >0$ be small enough. Assume that $0
< \lambda   \ll \alpha ^2 $. Then there exists a map $\Phi_\lambda: \mathbb T^2 \times A
\rightarrow \mathbb T^2 \times A$ with the following properties:
\begin{enumerate}
            \item $\Phi_\lambda$, depending analytically on $\lambda$, is a $C ^\infty-$diffeomorphism onto its image, close to the identity map in the $C ^\infty-$topology.
            \item For each $\xi \in A$, the invariant Lagrangian torus $\Lambda_\xi $ is sent, by $\Phi_\lambda$, to $ \Phi_\lambda (\Lambda_\xi) $ which is a Lagrangian
                torus, (close to $\Lambda_\xi$), denoted by $\Lambda_{\xi_ \lambda} $, and of the form $\Lambda_{\xi_\lambda}= \mathbb T^2
                \times \{\xi_\lambda \} $, with some $\xi_\lambda$ in a certain open subset
                $A_\lambda \subset A$, induced by $\Phi_\lambda$.

                Moreover, if $\xi \in A_{\alpha, d}$, then the torus $\Lambda_{\xi_ \lambda}$, with
                $\xi_\lambda$ in a certain subset $A_{\alpha, d, \lambda} \subset A_\lambda $,
                is still Diophantine $H_{p_\lambda}-$invariant torus, called a local KAM torus. The restricted map $\Phi_\lambda |_{\Lambda_\xi} $ on each Diophantine invariant Lagrangian torus $\Lambda_\xi$, with $\xi \in A_{\alpha, d}$,
                conjugates the Hamiltonian vector field
                $H_p|_{\Lambda_\xi}= H_p(x, \xi)$, given in \eqref{Hp}, to the Hamiltonian vector
                field $H_{p_\lambda}|_{\Lambda_{\xi_ \lambda}}$, i.e., $ (\Phi_\lambda|_{\Lambda_\xi})_* H_p = H_{p_\lambda}$.

                In particular, if $\xi \in \Gamma_a $, then the torus $\Lambda_{\xi_ \lambda}
                \subset p_\lambda ^{-1}(E) \cap \ \mathbb T^2 \times \Gamma_{a, \lambda}$, with a
                certain curve $\Gamma_{a, \lambda} \subset A_\lambda$. Moreover, for $\xi \in
                \Gamma_a \cap A_{\alpha, d}$, the Liouville measure of the complement of the
                union of the KAM tori $\Lambda_{\xi_ \lambda} $, in $p_\lambda ^{-1}(E)$, is
                of order $\mathcal{O}(\alpha)$.
\end{enumerate}
\end{theo}

Note that in the previous theorem, we can see $\xi_\lambda$ as a smooth function of $ \xi \in A$ and of the parameter $\lambda$.
From this theorem, we obtain a nowhere dense of large measure of KAM tori,
 on each of which the Hamiltonian flow of the perturbed system $p_\lambda$ is quasi-periodic of constant vector field.

% , that is $A_{\alpha, d, \lambda}$,

\subsection{General assumptions}   \label{gnrass}

For $h \ll \varepsilon = \mathcal{O}(h^\delta)$, with $0< \delta <1 $, and a small enough parameter $0< \lambda \ll 1$, let $P_{\varepsilon, \lambda}$ be a pseudodifferential operator depending smoothly on  $\varepsilon$ and $ \lambda$ such that
 \begin{equation} \label{a2a}
         P_{\lambda}:=P_{\varepsilon =0, \lambda} \quad \textrm{is formally selfadjoint}.
    \end{equation}
We assume that $P_{\varepsilon, \lambda}$ is classical of order $0$, that is $P_{\varepsilon, \lambda} \in  \Psi^0(m)(M)$ as in Section \ref{sec2.1}, here $m$ is an order function satisfying $m >1$. We assume the total symbol of $P_{\varepsilon, \lambda}$ is a holomorphic Hamiltonian in a tubular neighborhood of $T^*M$.

%%%%%%%%%%%%%%%%%%%%%%%%%%%%%%%%%%%%%%%%%%%%%%%%%%%%%%%%%%%%%%%%%%%%%%%%%%%
%%%%%%%%%%%%%%%%%%%%%%%%%%%%%%%%%%%%%%%%%%%%%%%%%%%%%%%%%%%%%%%%%%%%%%%%%%%%%%%%

Let $p_1$ be an analytic bounded function in a tubular neighborhood of $T^*M$, real on the real domain,
with $p_1(x, \xi) =  \mathcal O( m(Re(x,\xi))$ in the case when $M=\mathbb R^2$, and $p_1(x, \xi)=
 \mathcal O( \langle \xi \rangle ^m) $ in the manifold case.

Let $p_{\varepsilon, \lambda}$ be the semiclassical principal symbol of $P_{\varepsilon, \lambda}$.
In Taylor expansion in $\varepsilon$, it can be written in the form
 \beq  \label{p_lambda} p_{\varepsilon, \lambda}=p_\lambda +i \varepsilon q+ \mathcal O (\varepsilon
^2). \eeq
Here we assume that $ q $ is a bounded analytic real-valued function on $T^*M$, and $p_\lambda$ is a nearly integrable system of the form
 \beq   p_\lambda = p+ \lambda p_1, \eeq
with $p$ is a completely integrable system.

We note that $p_\lambda$ is the principal symbol of $P_{\lambda}$.

Let  $$p_{\varepsilon}=p_{\varepsilon, \  \lambda=0}= p +i \varepsilon q+ \mathcal O (\varepsilon
^2).$$
For some $E \in \R$, we assume some general dynamical assumptions on this energy level.
The first we assume the ellipticity condition at infinity as the following.

    When $M=\mathbb R^2$, the ellipticity condition at infinity is
                        \begin{equation}    \label{el con}
                            |p_{\varepsilon}(x,\xi) - E| \geq \frac{1}{C} m(Re(x,\xi)), \mid (x,\xi)\mid \geq C,
                        \end{equation}
     for some $ C> 0 $ large enough.

     When $ M $ is in manifold cases, the elipticity condition at infinity is
            \begin{equation} \label{el con2}
                            |p_{\varepsilon}(x,\xi) -E | \geq \frac{1}{C} \langle \xi \rangle ^m, (x,\xi) \in T^*M, \mid \xi \mid \geq C,
            \end{equation}
  for some $C>0$ large enough.

% Pho la discrete trong mot dai co epsilon
We note that by the elipticity condition at infinity, the spectrum of $P_{\varepsilon, \lambda}$ in a small neighborhood of $ E$ in $\mathbb C $ is discrete, when $h, \varepsilon$, and $\lambda$ are small enough. Moreover, this spectrum is contained in a horizontal band of size $\varepsilon$,
   $$    |\mathrm{Im} (z)| \leq \mathcal O(\varepsilon ).$$
%\label{band}

%%%%%%%%%%%%%%%%%%%%%%%%%%%%%%%%%%%%%%%%%%%%%%%%%%%%%%%%%%%%%%%%%%%%%%%%%%%%%%%

Because $ p $ is completely integrable, so there exists an integrable system
\begin{equation} \label{F}
F=(p,f): T^*M \rightarrow \mathbb R^2
\end{equation}
Then the space of regular leaves of $F$ is foliated by Liouville invariant tori by Theorem \ref{A-A}.
We assume also that
  \begin{equation}p^{-1}(E) \cap T^*M  \ \textrm{is connected}, \end{equation}
and the energy level $ E $ is regular for $ p $, i.e., $dp \neq 0$ on $p^{-1}(E) \cap T^*M$. We would like to notice that the level set $p^{-1}(E)$ is compact, due to the ellipticity condition at infinity \eqref{el con} or \eqref{el con2}

Then the energy space $p^{-1}(E)$ is decomposed into
a singular foliation:
 \begin{equation} \label{dk18}
 p^{-1}(E) \cap T^*M   = \bigcup_{a \in J} \Lambda_a ,\end{equation}
 where $ J $ is assumed to be a compact interval, or, more generally, a connected graph with a finite number of vertices and of edges.

We denote by $S$ the set of vertices.
For each $a \in J$, $\Lambda_a$ is a connected compact subset invariant with respect to $H_p$.
Moreover, if $a \in J\backslash S$, $\Lambda_a$ are the Liouville tori depending
analytically on $a$. These tori are regular leaves corresponding to regular values of $F$.
Each edge of $ J $ can be identified with a bounded interval of $ \mathbb R $ and we have therefore
a distance on $J$ in the natural way.

We denote $H_p$ the Hamiltonian vector field of $p$, defined by $\sigma (H_p, \cdot)= -dp (\cdot)$.
For each $a \in J$, we define a compact interval in $\mathbb R$:
                            \begin{equation} \label{Q vo cung}
                                    Q_\infty(a)=
                                     \big [ \lim_{T\rightarrow \infty} \inf_{\Lambda_a} Re \langle q \rangle _T,
                                      \lim_{T\rightarrow \infty} \sup_{\Lambda_a} Re \langle q \rangle _T\big],
                            \end{equation}
where $\langle q \rangle _T$, for $T>0$, is the symmetric average time $T$ of $q$ along the
$H_p-$flow, defined by
% \label{t-average}
                           $$
                                 \langle q \rangle _T= \frac{1}{T} \int_{-T/2}^{T/2} q \circ exp(t H_p) dt.
                           $$

For each torus $\Lambda_a$, with $ a  \in J\backslash S$, there exists an angle-action local chart $(V, \kappa)$  on an open neighborhood of $\Lambda_a$ as in \eqref{coor}
such that $\Lambda_a \simeq \Lambda_{\xi_a}$, $\xi_a \in A$, and $p$ becomes a function only of $\xi$,
           \begin{equation} \label{p}
                    p\circ \kappa^{-1} =p(\xi), \ \xi \in A.           \end{equation}
We define the function $\langle q \rangle $-the average of $ q $ on Liouville tori $\Lambda \subset V $,
with respect to the natural Liouville measure on $\Lambda$, denoted by $\langle q \rangle_{\Lambda} $, as the integral
                      \begin{equation} \label{moyenne de q}
                       \langle q \rangle_{\Lambda}= \int_{\Lambda}q .\end{equation}

\begin{rema}
 In the action-angle coordinates $(x,\xi)$ given by \eqref{coor}, we have
 \begin{equation} \label{moyenne2}
 \langle q \rangle_{\Lambda} =  \langle q \rangle_{\Lambda_ \xi}= \langle q \rangle (\xi)=
 \frac{1}{(2\pi)^2}\int_{\mathbb{T}^2}q(x,\xi)dx, \ \xi \in A.
 \end{equation}
In particular, $\langle q \rangle_{\Lambda_a}=\langle q \rangle(\xi_a)$.
\end{rema}

It is true that $\langle q \rangle_{\Lambda_a} $ depends analytically on $a \in J\backslash S$, and
we assume that it can be extended continuously on $J$. Furthermore, we assume that the function $a
\mapsto \langle q \rangle_{\Lambda_a} = \langle q \rangle(\xi_a)$ is not identically constant on
any connected component of $J\backslash S$, and that
\begin{equation}
\textrm{$dp(\xi)$ and $d \langle q \rangle (\xi)$ are $\mathbb R-$linearly independent at $\xi_a$.}
\end{equation}

% The frequency of $\Lambda$ (also of $\Lambda_{\xi}$) is defined by
% \begin{equation} \label{frequence} \omega(\xi)= \frac{\partial p}{\partial \xi} (\xi)= \big ( \frac{\partial p}{\partial \xi_1} (\xi),\frac{\partial p}{\partial \xi_2} (\xi) \big ), \ \xi \in A .\end{equation}

% Sometimes $\omega(\xi)$ is seen as an element of $\mathbb R$. \\
% Do tinh kha tich va dl A-A
It is true that the rotation number $a \mapsto  \rho (a):=\rho   (\xi_a)$, where $\rho$ is defined by \eqref{rho}, depends analytically on $a \in J\backslash S$. We assume moreover that
\begin{equation} \label{w}
\textrm{$ \rho$ is not identically constant on any connected
component of $J \backslash S$.}
\end{equation}

\begin{rema}[see Ref. \cite{Hitrik07}]  \label{rem1}
For $a \in J\backslash S$, if $\rho(a) \notin \mathbb{Q}$, that means the frequency $\omega(a)$
is non resonant, then the Hamiltonian flow of $p$  along the torus $\Lambda_a$ is ergodic. Hence the
limit of $\langle q \rangle _T$, when $ T \rightarrow \infty$ exists, and is
equals to $ \langle q \rangle_{\Lambda_a}$. Therefore we
have $$Q_\infty(a)= \{\langle q \rangle_{\Lambda_a}\}.$$
\end{rema}
%%%%%%%%%%%%%%%%%%%%%%%%%%%%%%%%%%%%%%%%%%%%%%%%%%%%%%%%%%%%%%%%%%%%%%%%%%%%%%%%%%%%%%%%%%%%%%
%%%%%%%%%%%%%%%%%%%%%%%%%%%%%%%%%%%%%%%%%%%%%%%%%%%%%%%%%%%%%%%%%%%%%%%%%%%%%%%%%%%%%%%%%%%%%%%%%%%

\subsection{Spectral asymptotics}   \label{end mono}

Let $P_{\varepsilon, \lambda}$ be the operator given in the last Section and we keep all notations as before.
Applying the spectral asymptotic theory Ref. \cite{Hitrik07}, we can give the asymptotics of eigenvalues of $P_{\varepsilon, \lambda}$ located in suitable small windows of the complex plane, which are associated to KAM tori, talked in Section \ref{kams}.

\begin{defi}  \label{dn bonnes valeurs}
        For some $\alpha>0 $ and some $d>0$, we define the set of good values associated with the energy level $E$, denoted by $\mathcal{G}(\alpha,d, E)$, obtained from $\cup_{a \in J}Q_\infty(a)$ by removing the following set of bad values $\mathcal{B}(\alpha,d, E)$:
\beas
\mathcal{B}(\alpha,d, E)& =  & \Bigg ( \bigcup_{dist(a,S) < \alpha} Q_\infty(a) \Bigg )
                                \bigcup \Bigg ( \bigcup_{a \in J\backslash S: \ \omega(\xi_a) \textrm{ is not }
                                    (\alpha,d)- \textrm{Diophantine}} Q_\infty(a)\Bigg )
                                          \\
                  &         &  \bigcup \Bigg( \bigcup_{a \in J\backslash S: \ |d \langle q \rangle(\xi_a) | < \alpha }
                                            Q_\infty(a) \Bigg )
 \bigcup \Bigg ( \bigcup_{a \in J\backslash S: \ |\rho'(a)| < \alpha } Q_\infty(a) \Bigg ) .
\eeas
\end{defi}

\begin{rema} \label{rem2} We have some remarks
         \begin{enumerate}[(i)] \label{peti mesure}
         \item When $d>0$ is kept fixed, the measure of the set of bad values
             $\mathcal{B}(\alpha,d, E)$ in $\cup_{a \in J}Q_\infty(a)$ (and in $\langle q
             \rangle_{\Lambda_a} (J) $) is $\mathcal O
             (\alpha)$, when $\alpha >0$ is small enough, provided that the measure of
                                    \begin{equation}  \label{a condition}
                                            \Bigg ( \bigcup_{a \in
                                             J\backslash S: \ \rho(a) \ \in \ \mathbb Q } Q_\infty(a)
                                             \Bigg ) \bigcup \Bigg ( \bigcup_{a \ \in \ S} Q_\infty(a)
                                            \Bigg )
                                    \end{equation}
             is sufficiently small, depending on $\alpha$ (see Ref. \cite{Hitrik07}).
             \item Let $G \in \mathcal{G}(\alpha,d, E) $ be a good value, then by Definition of
                 $\mathcal{B}(\alpha,d, E) $ and Remark \ref{rem1}, there are a finite number
                 of corresponding $(\alpha, d)-$Diophantine tori
                 $\Lambda_{a_1},\ldots,\Lambda_{a_L}$, with $L \in \mathbb N^*$ and $\{
                 a_1,\ldots,a_L\} \subset J\setminus S $, in the energy space $p^{-1}(E) \cap
                 T^*M$, such that the pre-image
                                         $$\langle q \rangle ^{-1}(G)= \{\Lambda_{a_1},\ldots, \Lambda_{a_L} \}.$$
             In this way, when $G $ varies in $\mathcal{G}(\alpha,d, E) $, we obtain a family of large measure
             of $(\alpha, d)-$Diophantine invariant tori  in the phase space
                 satisfying $ \{p= E, \langle q \rangle = G \}$.
         \end{enumerate}
\end{rema}

Now let $G \in  \mathcal{G}(\alpha,d, E)$ be a good value. As in Remark \ref{rem2}, there exists
$L$ elements in pre-image of $G$ by $\langle q \rangle$.

We shall assume that $L=1$. Then we can write
\begin{equation} \label{preimage}
     \langle q \rangle ^{-1}(G)=\Lambda_ a \subset  p^{-1}(E) \cap T^*M , \ a \in J\setminus S.
    \end{equation}
Note that this hypothesis can be achieved if we assume further that the map $(p, \langle q \rangle)$ has connected fibres.
We denote also $a=(E, G) $ and observe that the corresponding torus $\Lambda_a  \simeq \Lambda_{ \xi_a}$, with $\xi_a \in
\Gamma_a$, is a Diophantine torus.

Now we assume that $p$ is locally non-degenerate in a neighborhood of $\Lambda_a$, according Definition \ref{Kol}, and
moreover $0 < \lambda   \ll \alpha^2  $.

Then by Theorem \ref{kam}, there exists a nowhere dense
set of large measure of KAM tori close to $\Lambda_{ \xi_a}$, $\Lambda_{\xi_ \lambda} = \mathbb T^2 \times \{\xi_\lambda \}$, with $\xi_\lambda  \in A_{\alpha, d, \lambda}$, on each of which the $H_{p_\lambda}-$flow is quasi periodic of a Diophantine constant frequency, denoted by
$\omega_{\lambda}(\xi_\lambda)$. Therefore, over these KAM tori $\Lambda_{\xi_ \lambda}$, $p_\lambda$ becomes a function of only $\xi_\lambda$,
$$p_\lambda= p_\lambda (\xi_\lambda), \xi_\lambda  \in A_{\alpha, d, \lambda}. $$

% We assume that $p_\lambda$ extends to a smooth function of $\xi_\lambda$ on the whole domain  $A_\lambda$.

In particular, when $\xi \in \Gamma_a \cap A_{\alpha, d}$, then $\Lambda_{\xi_ \lambda} \subset
p_\lambda ^{-1}(E) \cap \ \mathbb T^2 \times  (\Gamma_{a, \lambda} \cap A_{\alpha, d, \lambda} )$, and with $\xi_ \lambda \in \Gamma_{a, \lambda} \cap A_{\alpha, d, \lambda} $.

By \eqref{moyenne de q}, we define locally an analytic function $\langle q \rangle
_{\Lambda_{\xi_ \lambda} }$ of $\xi_\lambda  \in A_\lambda$, denoted by $\langle q \rangle (\xi_
\lambda )$, obtained by averaging $q$ over the tori $\Lambda_{\xi_ \lambda} $.
Then we have, in $C^1-$sense in $\xi_\lambda  \in A_\lambda$,
 $$\langle q \rangle (\xi_ \lambda ) = \langle q \rangle _{\Lambda_{\xi_ \lambda} }   \rightarrow
 \langle q \rangle _{\Lambda_{\xi } }=  \langle q \rangle (\xi), \ \textrm{as} \ \lambda \rightarrow 0 .$$
%  \langle q \rangle _{\Lambda_{\xi_ \lambda |_\lambda=0 } }=
% that is close to $ \langle q \rangle (\xi_a)= \langle q \rangle _{\Lambda_{\xi_a } }  $. \\
Hence, with regard to the properties of the good value $G$, for every $\xi_\lambda \in \Gamma_{a, \lambda}$ and $\lambda$ small enough, it is true that
\beq
\label{dk1} \big | d \langle q \rangle (\xi_ \lambda ) \big | \geq \frac{\alpha}{2}. \eeq Notice that
we have also the differentials of $p_\lambda(\xi_ \lambda )$ and $\langle q \rangle (\xi_ \lambda
)$ in every $\xi_\lambda \in \Gamma_{a, \lambda} \cap A_{\alpha, d, \lambda} $ are $\mathbb R
-$linearly independent: \beq \label{dk2} \omega_{\lambda}(\xi_\lambda) \wedge d \langle q \rangle
(\xi_ \lambda ) \neq 0. \eeq
Let us define the set of (new) \textit{good values} for $P_{\varepsilon, \lambda}$,
\beq \label{K}
\mathcal{G_\lambda}(\alpha, E, G)= \bigcup_{\xi_\lambda \in \ \Gamma_{a, \lambda} \ \bigcap \ A_{\alpha, d,
\lambda}} \langle q \rangle (\xi_ \lambda ) = \bigcup \{\langle q \rangle (\xi_ \lambda ): \ \xi
\in \Gamma_a \cap A_{\alpha, d} \} .\eeq
For a fixed energy level $E$, near $G$ this set is a slight perturbation of the set $ \mathcal{G}(\alpha,d, E)$.
It is true that the measure of the complement of
$\mathcal{G_\lambda}(\alpha, E, G)$ in $\bigcup_{\xi_\lambda \in \Gamma_{a, \lambda}} \langle q
\rangle (\xi_ \lambda ) $ is of order $\mathcal{O}(\alpha)$, when $\alpha$ is small and
$d$ is kept fixed.

For any $K \in \mathcal{G_\lambda}(\alpha, E, G) $, then there exists a unique $\xi \in \Gamma_a \cap A_{\alpha,
d} $ such that $\langle q \rangle (\xi_ \lambda ) = \langle q \rangle _{\Lambda_{\xi_ \lambda} }= K$, and the corresponding KAM torus $\Lambda_{\xi_ \lambda} \subset p_\lambda ^{-1}(E)$ is still Diophantine $H_{p_\lambda}-$invariant Lagrangian torus. Moreover, the $H_{p_\lambda}-$flow on $\Lambda_{\xi_
\lambda} $ is quasi-periodic of the Diophantine constant frequency $\omega_{\lambda}(\xi_\lambda)$,
satisfying \eqref{dk1} and \eqref{dk2}. Therefore, the general assumptions on the dynamic of $H_{p_\lambda}$ allow to carry out the Birkhoff normalization for $P_{\varepsilon, \lambda}$ near $\Lambda_{\xi_ \lambda} $, that leads to spectral asymptotic results from $h-$depending complex window.

We define in the horizontal band of size
$\varepsilon$ a suitable window of size
$\mathcal{O}(h^\delta) \times \mathcal{O}(\varepsilon h^\delta)$, around the \textit{good center $E+ i\varepsilon K$},
called \textit{good rectangle},
 \begin{equation} \label{cua so}
                        R^{(E, K)}(\varepsilon,h)
                                    = (E+i \varepsilon \ K)+  \Big[-\frac{h^\delta}{\mathcal{O}(1)},\frac{h^\delta}{\mathcal{O}(1)} \Big]
                                    +i \varepsilon \Big [ -\frac{h^\delta}{\mathcal{O}(1)}, + \frac{h^\delta}{\mathcal{O}(1)} \Big ].
 \end{equation}

\begin{rema}
The crucial idea for the Birkhoff normalization (Refs. \cite{Hitrik07}, \cite{Charles08}, and \cite{Pov}) is to use consecutively canonical transformations near a Diophantine invariant torus (like $\Lambda_{\xi_ \lambda}$), to conjugate the operator $P_{\varepsilon, \lambda}$ to a new operator, whose total symbol is independent of angle variables $x$, and homogeneous polynomial to high order in $ (\xi, h, \varepsilon)$. %This conjugation, with the help of the Egorov Theorem \cite{Egorov69}, is by means of analytic Fourier integral operators \cite{Duis11}.
The dynamical assumptions and the Diophantine condition (for the torus $\Lambda_{\xi_ \lambda} $) is indispensable for this construction.
Therefore we obtains spectral asymptotics for the operator $P_{\varepsilon, \lambda}$, that are similar results in integrable cases of Ref. \cite{QS16.1}.
\end{rema}

%%%%%%%%%%%%%%%%%%%%%%%%%%%%%%%%%%%%%%%%%%%%%%%%%%%%%%%%%%%%%%%%%%%%%%%%%%%%%%%%%%%%%%%
We recall here some notations used for the statement of results.

Let $\Lambda$ be an invariant Lagrangian torus and suppose that $\kappa_\infty$ is an
action-angle local chart such that $\Lambda \simeq \Lambda_{\xi}$ (like \eqref{coor}). The fundamental cycles $(\gamma_1,\gamma_2)$ of $\Lambda$ associated to $\kappa_\infty$ are defined by
            $$ \gamma_j= \kappa_\infty^{-1} ( \{ (x, \xi) \in T^* \mathbb{T}^2: x_j=0 \}) , \ j=1,2. $$
Then the action integrals of these cycles are defined by
\begin{equation} \label{act}
    S= (S_1, S_2) = \left( \int_{\gamma_1} \theta , \int_{\gamma_2} \theta  \right) \in \R^2,
\end{equation}
where $\theta$ is the Liouville $1-$form on $T^*M$.

\begin{defi}[Refs. \cite{lectureColin}, and \cite{Rob93}]   \label{ind}
% xem dinh ngia grassmanienne pag 16 quyen Yves.
Let $ E $ be a symplectic space and let $ \Lambda (E) $ be his Lagrangian Grassmannian (which is
set of all Lagrangian subspaces of $ E $). We consider a bundle $ B $ in $ E $ over the circle or a
compact interval provided with a Lagrangian subbundle called vertical. Let $ \lambda (t) $ be a
section of $ \Lambda (B) $ which is transverse to the vertical edges of the interval in the case
where the base is an interval. The Maslov index of $ \lambda (t) $ is the intersection number of
this curve with the singular cycle of Lagrangians which do not cut transversely the vertical
subbundle.
\end{defi}

The Maslov index appears in the statement of spectral asymptotic results. In our work it is treated as a standard constant.

%%%%%%%%%%%%%%%%%%%%%%%%%%%%%%%%%%%%%%%%%%%%%%%%%%%%%%%%%%%%%%%%%%%%%%%%%

Using a result from Ref. \cite{Hitrik07}, Section 7.3, for the energy level $0$ , we can have the following result for any energy level $E$ in the range of $p$.
\begin{theo}  \label{v}
For $E \in \R$. Let $P_{\varepsilon, \lambda}$ be an operator satisfying all general assumptions of Section \ref{gnrass}, \eqref{a2a}- \eqref{w}.
Let $G$ be a good value according Def. \ref{dn bonnes valeurs} with $\alpha > 0$ small enough and $d>0$ fixed, such that the condition \eqref{preimage} is true.
We assume that $ \lambda   \ll \alpha^2$, and $p$ is locally non-degenerate, according Definition \ref{Kol}. Then we can define the set of good values $\mathcal{G_\lambda}(\alpha, E, G)$ by \eqref{K}.

For each $K \in \mathcal{G_\lambda}(\alpha, E, G) $, there exists a unique $H_{p_\lambda}-$ invariant KAM torus $\Lambda_{\lambda}=  p_\lambda ^{-1}(E) \cap \langle q \rangle ^{-1}(K) $, and a canonical transformation
\begin{equation} \label{kappa vc}
\kappa_{\infty}=(x, \xi): \textrm{neigh} \ (\Lambda_{\lambda}, T^*M)\rightarrow
\textrm{neigh} \ (\Lambda_{\xi_\lambda}, T^*\mathbb T^2) ,
\end{equation}
mapping $\Lambda_{\lambda}$ to $\Lambda_{\xi_\lambda}= \mathbb T^2 \times \{\xi_\lambda \} $, such that
 \begin{equation} \label{lea. sym}
  p_\lambda \circ \kappa_{\infty}^{-1}= p_\lambda (x, \xi)=  p_{\lambda, \infty }(\xi) + \mathcal{O}(\xi-\xi_\lambda )^\infty,
 \end{equation}
where $p_{\lambda, \infty }$ is a smooth function depending only on $\xi$, and admits the Taylor expansion at $\xi_\lambda$ of the form
\begin{equation} \label{Taylor}
p_{\lambda, \infty }(\xi)= E+ \omega _{\lambda}(\xi_\lambda)\cdot (\xi-\xi_\lambda ) + \mathcal{O}(\xi-\xi_\lambda )^2.
\end{equation}
Let $\eta \in \mathbb Z^2$ be the Maslov indices (see Def. \ref{ind}) and $S \in \mathbb R^2$ be the action integrals, defined by \eqref{act} of
the fundamental cycles of $\Lambda_{\xi_ \lambda} $, suitably with the canonical transformation \eqref{kappa vc}.

Then all eigenvalues $\mu$ of
$P_{\varepsilon, \lambda}$ in the good rectangle $R^{(E, K)}(\varepsilon,h)$, defined by \eqref{cua so}, are given as the image of a portion of $h \mathbb Z^2$, modulo $\mathcal O(h^\infty)$, by a smooth function denoted by $P_\lambda (\xi, \varepsilon;h)$ of $\xi$ in a neighborhood of $(\xi_\lambda, \mathbb R^2)$ and $\varepsilon, h$ in neighborhoods of $(0, \mathbb R)$:
 \begin{equation} \label{eigenvalues-with lambda}
   \sigma(P_{\varepsilon, \lambda}) \cap R^{(E, K)}(\varepsilon,h) \ni \mu =
   P_\lambda \Big( \xi_\lambda+ h(k-\frac{\eta}{4})-\frac{S}{2 \pi};\varepsilon, h\Big) + \mathcal O(h^\infty), \  k \in \mathbb Z^2.
  \end{equation}
Moreover, $P_\lambda $ is a real valued function for $\varepsilon =0 $, admits a polynomial asymptotic expansion in $(\xi, \varepsilon, h)$ for the $C^\infty-$ topology of the form
    \begin{equation} \label{symbole normal}
        P_\lambda ( \xi; \varepsilon, h) \sim    \sum_{\alpha, j,k} C_{\alpha j k} \ \xi^\alpha \varepsilon^ j h^k.
     \end{equation}
In particular, the $h-$leading term of $P_\lambda$ is of the form
\beq \label{prin normal-lambda} p_{0, \lambda}(\xi, \varepsilon) =
p_{\lambda, \infty }(\xi)+ i \varepsilon \langle q \rangle (\xi) + \mathcal O(\varepsilon ^2), \eeq
where $p_{\lambda, \infty }(\xi)$ is given by \eqref{lea. sym}-\eqref{Taylor}, and $ \langle q \rangle (\xi)$ is the expression of the average of $q$ over a torus $\Lambda_\xi$ close to $\Lambda_{\xi_\lambda} $, in the previous coordinates, see \eqref{moyenne2}.
\end{theo}

\begin{rema} \label{lov}
We notice that in the theorem the canonical transformation $\kappa_{\infty} $ in \eqref{kappa vc} satisfying \eqref{lea. sym} and \eqref{Taylor}, is locally valid for all KAM tori $\Lambda_{\xi_\lambda}$ in a neighborhood of the Diophantine Liouville $\Lambda_a$ given in \eqref{preimage}, here $a=(E,G)$.
\end{rema}

\begin{rema} \label{diffeo}
We can show that the function $P_\lambda $ in the above theorem is a local diffeomorphism from a neighborhood of $(\xi_a, \mathbb R^2)$
into its image, that is in a $\mathcal O(\varepsilon)-$ horizontal band (see Ref. \cite{QS14}). Therefore the eigenvalues of $P_{\varepsilon, \lambda}$ in a good rectangle form a deformed $h-$\textit{lattice}.
Moreover, the lattice has a horizontal spacing $ \mathcal O(h)$
and a vertical spacing $ \mathcal O(\varepsilon h)$.
\end{rema}

%%%%%%%%%%%%%%%%%%%%%%%%%%%%%%%%%%%%%%%%%%%%%%%%%%%%%%%%%%%%%%%%%%%%%%
%%%%%%%%%%%%%%%%%%%%%%%%%%%%%%%%%%%%%%%%%%%%%%%%%%%%%%%%%%%%%%%%%%%%%%%%%%%%%%%%%%

\subsection{Asymptotic pseudo-lattice structure of the spectrum}
% of $P_{\varepsilon, \lambda}$
Assume that $p$ is a completely integrable Hamiltonian system. Let $U$ be an open subset with compact closure of regular values of the moment map $F=(p,f)$ given in \eqref{F}, and let $X= F^{-1}(U)$.
Then we can cover $X$ by an atlas of angle-action charts $\{( V^c, \kappa^c ) \}_{c \in U}$ by Theorem \ref{A-A}.

\begin{defi} \label{gnd}
We say that $p$ is globally non-degenerate (on $X$) if for such an atlas of $X$, $p$ is locally non-degenerate on every $V^c, c \in U$, according Definition \ref{Kol}.
\end{defi}

With general assumption on the dynamics as in Section \ref{gnrass} and the globally non-degeneracy, we shall show the main result that the spectrum of $P_{\varepsilon, \lambda}$ on the domain $U(\varepsilon)$ satisfies all hypotheses of an asymptotic pseudo-lattice, according Definition \ref{pseu-lattice}.

\begin{theo} \label{mtheo2}
Let $p$ be a completely integrable Hamiltonian and the set $U$ as above. We assume that $p$ is globally non-degenerate, according Def. \ref{gnd}.
For $h \ll \varepsilon = \mathcal{O}(h^\delta)$, with $0< \delta <1 $, and a small enough parameter $0< \lambda \ll 1$,
we consider an operator $P_{\varepsilon, \lambda}$ satisfying all general assumptions of Section \ref{gnrass}, \eqref{a2a}- \eqref{w}, for any energy level $E$ in the projection of $U$ on horizontal axis.
Let $d >0$ fixed, $\alpha >0$ small enough, then we assume that $0 < \lambda   \ll \alpha ^2 $.
We assume further that the map $(p, \langle q \rangle )$ on $X $ has connected fibers, where $\langle q \rangle $ is the average of $q$ on tori, given in \eqref{moyenne de q}.

We define the set $U(\varepsilon)$ as in \eqref{elp}. Then $(\sigma( P_{\varepsilon, \lambda}), U(\varepsilon) )$ is an asymptotic pseudo-lattice.
\end{theo}

\begin{proof}[Proof of Theorem \ref{mtheo2}]

Let an any point $c \in U$. For any point $a=(E,G) $ in a small enough neighborhood $U^c $ of $c$ such that $G$ is a good value, i.e., $G \in \mathcal{G}(\alpha,d,E)$, according Def. \ref{dn bonnes valeurs}. With hypotheses of the theorem, the condition \eqref{preimage} holds ($L=1$).
For any $K \in \mathcal{G_\lambda}(\alpha, E, G)$ such that $(E,K) \in  U^c$.
Then there exists an unique KAM torus $\Lambda_{\xi_
\lambda}=  p_\lambda ^{-1}(E) \cap \langle q \rangle ^{-1}(K) $ of $H_{p_\lambda}-$flow, and a canonical transformation $ \kappa_{\infty} $ near $\Lambda_{\xi_\lambda}$ as in \eqref{kappa vc}.
It is clear here that all hypotheses of Theorem \ref{v} are satisfied and therefore we can have asymptotic expansions for
all eigenvalues $\mu$ of $P_{\varepsilon, \lambda}$ in a good rectangle $R^{(E, K)}(\varepsilon,h)$ of the from \eqref{cua so}, given by \eqref{eigenvalues-with lambda}:
$$
   \mu =
   P_\lambda \Big( \xi_\lambda+ h(k-\frac{\eta}{4})-\frac{S}{2 \pi};\varepsilon, h\Big) + \mathcal O(h^\infty), \  k \in \mathbb Z^2.
$$

We remark also a classical result that
 \begin{equation}  \label{pt action}
                             \frac{S}{2 \pi}- \xi_\lambda := \tau_c \in \mathbb R^2,
 \end{equation}
 is locally constant in $c \in U$ (see Ref. \cite{QS14}).
Then from Remark \ref{diffeo}, the equation \eqref{eigenvalues-with lambda}, with the aid of \eqref{prin normal-lambda}, provides a smooth local diffeomorphism near $E+i \varepsilon K \in \mathbb C$, denoted by $f_\lambda$, that sends $\mu$ to $hk \in h \mathbb Z^2 $, modulo $\mathcal O(h^\infty)$, of the form
 \begin{eqnarray}  \label{new hk}
              f_\lambda=f_\lambda(\mu, \varepsilon;h) & =&  \frac{S}{2 \pi}- \xi_\lambda + h \frac{\eta}{4}+  P_\lambda^{-1}(\mu)
                  \\
                      &= & \tau_c + h \frac{\eta}{4}+ P_\lambda^{-1}(\mu)
                       \nonumber \\
                \sigma(P_{\varepsilon, \lambda}) \cap R^{(E, K)}(\varepsilon,h) \ni  \mu & \mapsto &
                        f_\lambda(\mu,\varepsilon; h) \in h \mathbb Z^2 +\mathcal O(h^\infty).  \nonumber
 \end{eqnarray}

 Let $\widetilde{f_\lambda}= f_\lambda \circ \chi$, where $\chi$ given by \eqref{chi}. We have
\begin{equation}    \label{f tilde}
\widetilde{f_\lambda}=  \tau_c + h \frac{\eta}{4}+ P_\lambda^{-1} \circ \chi .
\end{equation}
Due to the fact that $P_\lambda$ admits an asymptotic expansion in $(\xi, \varepsilon, h)$ of the form \eqref{symbole normal} with $h-$leading term \eqref{prin normal-lambda}, then $\widetilde{f_\lambda}$ admits an
asymptotic expansion in $(\xi, \varepsilon,\frac{h}{\varepsilon})$, by Proposition \ref{D.A inverse} (see below).
In the reduced form, the asymptotic expansion in $(\varepsilon,\frac{h}{\varepsilon})$ of $\widetilde{f_\lambda}$, for the $C ^\infty-$topology in a neighborhood of $E+i K$, satisfies $$\widetilde{f}_{\lambda}= \widetilde{f}_{ \lambda,  0} + \mathcal
O(\varepsilon, \frac{h}{\varepsilon}), $$ uniformly for $h, \varepsilon$ small and $h \ll
\varepsilon $.
Here $\widetilde{f}_{\lambda,0}$ is the leading term of $\widetilde{f}_{\lambda}$, and
\begin{equation}   \label{z}
  \widetilde{f}_{\lambda,  0}= \tau_c+ \psi^{-1},
\end{equation}
with \begin{equation} \label{psi}
 \psi(\xi) := ( p_{\lambda , \infty } (\xi) , \langle q \rangle (\xi)  ),
\end{equation}
for $\xi$ in a neighborhood of $(\xi_\lambda, \mathbb R^2)$.

The Remark \ref{lov} allows us to confirm that the leading term $\widetilde{f_\lambda}_0$ is a local diffeomorphism, well defined on
$U^c$. It is locally valid for any good rectangle $R^{(E, K)}(\varepsilon,h) \subset U^c ( \varepsilon )$.
Hence $f_\lambda$ is exactly a $h-$chart of $\sigma(P_{\varepsilon, \lambda})$ on a good rectangle.

We note that the set of good values is of large measure. Then a family of such these $h-$charts with locally leading term satisfying \eqref{z} forms a local pseudo-chart for $\sigma(P_{\varepsilon, \lambda})$ on $U^c(\varepsilon )$.
The above construction ensures that $(\sigma(P_\varepsilon), U(\varepsilon) )$ is an asymptotic pseudo-lattice.
\end{proof}

%%%%%%%%%%%%%%%%%%%%%%%%%%%%%%%%%%%%%%%%%%%%%%%%%%%%%%%%%%%%%%%%%%%%%%%%%%%%%%%%

\begin{prop}(Ref. \cite{QS14})  \label{D.A inverse}
            Let $\widehat{P}= \widehat{P}(\xi; X)$ be a complex-valued smooth function of $\xi$
            near $0 \in \mathbb R^2$ and $X$ near $0 \in \mathbb R^n$.
            Assume that $\widehat{P}$ admits an asymptotic expansion in $X$ near $0$ of the form
                $$\widehat{P}(\xi; X) \sim  \sum_\alpha C_\alpha(\xi) X^\alpha,$$
            with $C_\alpha(\xi)$ are smooth functions and $C_0(\xi):=\widehat{P}_0(\xi) $ is a local diffeomorphism near $\xi=0$.

            Then, for $\mid X \mid $ small enough, $\widehat{P}$ is also a smooth local diffeomorphism near $\xi=0$ and its inverse admits an asymptotic expansion in $X$ near $0 $ whose the first term is
            $(\widehat{P}_0)^{-1}$.
\end{prop}

%%%%%%%%%%%%%%%%%%%%%%%%%%%%%%%%%%%%%%%%%%%%%%%%%%%%%%%%%%%%%%%%%%%%%%%%%%%%%%%%%%%%%%%%%%
Theorem \ref{mtheo2} allows us to define the monodromy of small non-selfadjoint perturbed operators as the monodromy of asymptotic pseudo-lattices. Then we have,
\begin{defi} \label{qism}
 We suppose that $P_{\varepsilon, \lambda}$ is an operator and $U$ is a set as in Theorem \ref{mtheo2}. Then we define the spectral monodromy of $P_{\varepsilon, \lambda}$, denoted by $[\mathcal M_{sp}(P_{\varepsilon, \lambda})]$, as the monodromy of the asymptotic pseudo-lattice $(\sigma( P_{\varepsilon, \lambda}), U(\varepsilon) )$, discussed in Section \ref{mapl}.
\end{defi}

The class $[\mathcal M_{sp}(P_{\varepsilon, \lambda})] \in \check{H}^1(U, GL(2, \mathbb Z) )$, defined by the $1$-cocycle $\{M_{ij} \in GL(2, \mathbb Z) \}_{i, j \in \mathcal{J} } $ of transition maps between local pseudo-charts, given in \eqref{transition-pseu}.
We note that  by the discreteness of the integer group $GL(2, \mathbb Z) $, this spectral monodromy does not depend on $\lambda$ small enough and the classical parameters $h, \varepsilon $.
 Moreover, with the help of \eqref{z} and \eqref{psi} we have
\begin{equation}  \label{d}
M_{ij}= d \big [ \widetilde{f}^{i}_{\lambda,  0} \circ (\widetilde{f}_{ \lambda,  0} ^{j})^{-1} \big ]= d \big [ (\psi^i)^{-1} \circ \psi_j \big ], \ i, j \in \mathcal{J}.
\end{equation}

%%%%%%%%%%%%%%%%%%%%%%%%%%%%%%%%%%%%%%%%%%%%%%%%%%%%%%%%%%%%%%%%%%%%%%%%%%%%%%%%%%%%%%
%%%%%%%%%%%%%%%%%%%%%%%%%%%%%%%%%%%%%%%%%%%%%%%%%%%%%%%%%%%%%%%%%%%%%%%%%%%%%%%%%%%%%%

%The spectral monodromy is well defined for small non-selfadjoint perturbations of a selfadjoint classical operators in dimension $2$, admitting a quasi-integrable unperturbed principal symbol, as discussed in Sec. \ref{end mono}.

On the other hand, for a nearly integrable system, there exists a geometric invariant for its invariant KAM tori, defined by Broer and co-worker in Ref. \cite{Broer07}. We shall call the KAM monodromy.
We discuss now the relationship between the spectral monodromy and the KAM monodromy.

The result of Ref. \cite{Broer07} showed that there exists a nowhere dense subset of $H_{p_\lambda}-$invariant KAM tori in $X$ of the form $\bigcup_{a \in C} \Lambda_{a, \lambda},$ where $C \subset U $ is a nowhere dense subset of large measure when $\lambda $ small enough.
Moreover, the measure of KAM tori in $X$ is large,
$$\mu(X \setminus \bigcup_{a \in C} \Lambda_{a, \lambda})  = \mathcal O( \alpha ) \mu (X). $$
The KAM monodromy is defined as the obstruction against global triviality of a $\mathbb Z^2-$bundle over the base $U$, denoted by $\mathscr F$, which is an extension from the $\mathbb Z^2-$bundle $H_1(\Lambda_{a, \lambda}, \mathbb Z) \rightarrow a \in C$.
Here $H_1(\Lambda_{a, \lambda}, \mathbb Z)$ is the first homology group of $\Lambda_{a, \lambda}$.

On the KAM tori $\Lambda_{ \lambda}$ locally embedded in a canonical transformation as in \eqref{kappa vc} such that  $\Lambda_{\lambda} \simeq \Lambda_{\xi_\lambda}$, we have locally $ (p_{\lambda}( \xi_\lambda), \langle q \rangle ( \xi_\lambda ) )=\psi(\xi_\lambda)$.
It is classical, see Ref. \cite{Duis80}, that the (periodic) Hamiltonian flows of $p_\lambda$ and $\langle q \rangle$ from a point on $\Lambda_{ \lambda}$ form a basis of $H_1(\Lambda_{\lambda}, \mathbb Z)$.
Therefore the transition maps between overlapping local trivializations of the bundle $\mathscr F$ are

$$\{ {}^t \big ( d \big [ (\psi^i)^{-1} \circ \psi_j \big ] \big ) ^{-1} \}, \ i, j \in \mathcal{J}.$$
By this result and \eqref{d}, we can state that the spectral monodromy of $P_{\varepsilon, \lambda}$  allows to recover completely the KAM monodromy. More precisely we have.
\begin{theo} \label{reco}
The spectral monodromy of $P_{\varepsilon, \lambda}$, defined in Def. \ref{qism}, is the adjoint of the KAM monodromy of the nearly integrable system $p_{ \lambda}$.
\end{theo}

In conclusion, the spectral monodromy of small non-selfadjoint perturbations of a selfadjoint operator admitting a principal symbol that is a nearly integrable and globally non-degenerate system, is well defined directly from its spectrum, independent of small perturbations.
Moreover, the spectral monodromy allows to recover the monodromy of KAM invariant tori of the underlying nearly integrable system.

% Truong hop dac biet khi perturbation is still integrable, then the same classical

% voi d co dinh, tinh chat Uniformly Dio la ko phu thuoc vao local coord, xem Broer...

% see also 37 for the classical Birkhoff construction in the non-resonant case, near a stable equilibrium point.

% AP de dn monodromy for classical resonances.

\section*{Acknowledgement} This work was completed during my visit to the Vietnam Institute for Advanced Study in Mathematics (VIASM). We would like to thank the Institute for its support and hospitality. We would like also to thank the Vietnam National University of Agriculture, who gave me a good opportunity to develop my research.

%%%%%%%%%%%%%%%%%%%%%%%%%%%%%%%%%%%%%%%%%%%%%%%%%%%%%%%%%%%%%%%%%%%%%%%%%%%%%%%%%%%%%%%%%%%%%%%%%%%%%%%%%%%%%%%%%%%%%%%%%%%%%%%%%%%%%%%%%%%%%%%%%%%%%%%%%%%%%%%%%%%%%%%%%%%%%%%%%%%
%%%%%%%%%%%%%%%%%%%%%%%%%%%%%%%%%%%%%%%%%%%%%%%%%%%%%%%%%%%%%%%%%%%%%%%%%%%%%%%%%%%%%%%%%%%%%%%%%%%%%%%%%%%%%%%%%%%%%%%%%%%%%%%%%%%%%%%%%%%%%%%%%%%%%%%%%%%%%%%%%%%%%%%%%%%%%%%%%%%
%%%%%%%%%%%%%%%%%%%%%%%%%%%%%%%%%%%%%%%%%%%%%%%%%%%%%%%%%%%%%%%%%%%%%%%%%%%%%%%%%%%%%%%%%%%%%%%%%%%%%%%%%%%%%%%%%%%%%%%%%%%%%%%%%%%%%%%%%%%%%%%%%%%%%%%%%%%%%%%%%%%%%%%%%%%%%%%%%%%

\bigskip

\end{document}